\documentclass[letterpaper, 10 pt, conference]{ieeeconf}

\usepackage{amsmath,amsthm,amssymb}
\usepackage{graphics}
\usepackage{hyperref}
\usepackage{bigints}
\usepackage{etoolbox}
\usepackage{optidef}
\usepackage[utf8]{inputenc}
\usepackage{caption,subcaption}
\usepackage{pgfplots}

\pgfplotsset{compat=newest}
\usepgfplotslibrary{groupplots}
\usepgfplotslibrary{dateplot}
\appto\UrlBreaks{\do\-}

\newcommand{\ie}{\textit{i}.\textit{e}., }

\IEEEoverridecommandlockouts  
 \usepackage{stackengine}
\usepackage{optidef}
\pgfplotsset{every tick label/.append style={font=\footnotesize}}
\usepackage{booktabs}
\usepackage{xcolor}
\title{\LARGE\bf{The Dirichlet Mechanism for Differential Privacy on the Unit Simplex}}
\author{Parham Gohari, Bo Wu, Matthew Hale, Ufuk Topcu
\thanks{Parham Gohari is with the Department of Electrical and Computer Engineering, University of Texas at Austin, Austin, TX. Bo Wu  and Ufuk Topcu are with the Department of Aerospace Engineering and Engineering Mechanics, and the Oden Institute for Computational Engineering and Sciences, University of Texas at Austin, Austin, TX. email: {\tt\small $\{$pgohari, bwu3, utopcu$\}$@utexas.edu}. Matthew Hale is with the Department of Mechanical and Aerospace Engineering at the University of Florida, Gainesville, FL. email: {\tt\small matthewhale@ufl.edu}}}
\newtheorem{theorem}{Theorem}

\newtheorem{lemma}{Lemma}
\newtheorem{proposition}{Proposition}
\newtheorem{definition}{Definition}

\theoremstyle{remark}
\newtheorem*{remark}{Remark}

\begin{document}
\maketitle
\begin{abstract}
As members of a network share more information with each other and network providers, sensitive data leakage raises privacy concerns. 
To address this need for a class of problems, 
we introduce a novel mechanism that privatizes vectors belonging to the unit simplex. Such vectors can be seen in many applications, such as privatizing a decision-making 
policy in a Markov decision process. 
We use differential privacy as the underlying mathematical framework for these developments. 
The introduced mechanism is a probabilistic mapping that maps a vector within the unit simplex to the same domain according to a Dirichlet distribution. We find the mechanism well-suited for inputs within the unit simplex because it always returns a privatized output that is also in the unit simplex. Therefore, no further projection back onto the unit simplex is required. We verify the privacy guarantees of the mechanism for two cases, namely, identity queries and average queries. In the former case, we derive expressions for the differential privacy level of privatizing a single vector within the unit simplex. In the latter case, we study the mechanism for privatizing the average of a collection of vectors, each 
of which is in the unit simplex. We establish a trade-off between the strength of privacy and the variance of the mechanism output, and we introduce a parameter to balance the trade-off between them.
Numerical results illustrate these developments. 
\end{abstract}

\section{Introduction}
In many decision-making problems, a policy-maker forms a control policy based on data collected from the individuals in a network. %Often times, 
The gathered data often contains sensitive information, which raises privacy concerns~\cite{hart1992nonintrusive}. 
In some applications, privatizing sensitive data has been achieved
by adding carefully calibrated noise to sensitive data and functions thereof~\cite{7833044,7577745,8431397}. These noise-additive approaches are well-suited to some classes of numerical data, though
sensitive data may take a form ill-suited to them. 
For example, developments in~\cite{8814723} explored symbolic control systems in which additive
noise cannot be meaningfully implemented. 

In this work, we privatize data inputs that belong to the unit simplex, \ie the set of vectors with non-negative entries that sum to one. Such vectors are seen in many decision-making problems. For example, in Markov decision processes (MDPs), the goal is to find a total-reward-maximizing policy \cite{puterman2014markov,sutton1998introduction}. %In certain cases, deterministic policies that map each MDP state to a single action are sufficient to guarantee optimally  \cite{sutton1998introduction,puterman2014markov}. 
In certain cases, it is shown that the optimal policy is a randomized function that maps from the MDP's states to a probability distribution on the set of actions available at that 
state, see, \emph{e.g.},~\cite{8735817,wu2019switched,chatterjee2006markov}. 
Finite action sets give rise to discrete, finitely supported probability distributions, which
can be formalized as vectors with non-negative entries summing to one. 
Policies of this kind arise in applications such as autonomous driving~\cite{brechtel11} and
the smart power grid~\cite{misra13}, and revealing them can therefore reveal individuals' behaviors.
Thus, there is a need to privatize such policies, and this use represents one application of
privatizing sensitive data in the unit simplex. Existing noise-additive approaches will not, in
general, produce a privatized vector in the unit simplex, and we therefore propose a new
approach to privacy for this context. 

In this paper, we use differential privacy as the underlying mathematical framework for privacy. Differential privacy, first introduced in~\cite{dwork2006calibrating}, is designed to
protect the exact values of sensitive pieces of data, while preserving their usefulness in aggregate
statistical analyses. 
Two desirable properties of differential privacy are (i) that it is immune to 
post-processing~\cite{dwork2014algorithmic},
in the sense that arbitrary post-hoc transformations of privatized data do not weaken its
privacy guarantees, and (ii) that it is robust to side information, in that gaining additional
information about data-producing entities does not weaken its privacy guarantees by 
much~\cite{kasiviswanathan2014semantics}. 
As a result, differential privacy has been frequently used as the mathematical formulation of privacy in 
both computer science and, more recently, in control 
theory~\cite{nissim2018privacy,cortes2016differential,han2018privacy}. 

As the main contribution of this paper, we introduce a mechanism that privatizes a vector within the unit simplex. A mechanism is a probabilistic mapping from some pre-defined domain to a pre-defined range, 
and a mechanism is used to privatize sensitive data. 
This paper develops a novel mechanism using the Dirichlet distribution, and we therefore 
call it the Dirichlet mechanism. 
The Dirichlet distribution is a multivariate distribution supported on the unit simplex,
which makes it a natural choice for this setting because its outputs are always elements of
the unit simplex. 

In our developments, we use probabilistic differential privacy, which is known to
imply that the conventional form of differential privacy also holds~\cite{gotz11}. 
Then, we show that the Dirichlet mechanism satisfies probabilistic differential privacy for
identity queries. By an identity query, we mean privatizing a single vector within the unit simplex. In the course of proving these privacy guarantees, based on the assumptions we provide, we prove the log-concavity of the cumulative distribution function of a Dirichlet distribution. The proof that we present may be of independent interest in ongoing research on convexity analysis of special functions such 
as~\cite{karp2016normalized}. 
In this vein, we prove a generalization of Theorem 6 of~\cite{prekopa1972} which has been used in later works for stochastic programming~\cite{Prekopa71logarithmicconcave}.

Beyond identity queries, we further show that the Dirichlet mechanism is differentially private for average queries, in which we privatize the average of a collection of vectors, each within the unit simplex. 
We derive analytic expressions for privacy levels of the averaging case, and %we leverage
%these results to 
show that the Dirichlet mechanism provides privacy protections 
whose strength increases with the number of vectors being averaged. 

Following the convention in the differential privacy literature, we also analyze the accuracy of the output of the mechanism~\cite{dwork2014algorithmic,4389483}. 
In particular, we evaluate the accuracy of the Dirichlet mechanism in terms of the expected value and the variance of its outputs. Similar to additive noise methods, the Dirichlet mechanism output has the same expected value as its input, which implies that its privatized outputs obey a distribution centered on the underlying sensitive data. We show that there exists a trade-off between the privacy and the variance of the output of the mechanism. The derived expression for the output variance shows how to tune the worst-case variance by scaling the input by a parameter that we introduce in the mechanism
definition. 

We emphasize that additive noise privacy mechanisms are ill-suited to privacy on the unit simplex
because they add noise of infinite support. As a result, such mechanisms will output a vector that
does not belong to the unit simplex; attempting to normalize the noise would result in its distribution
not being one known to provide differential privacy. It is for these reasons that we develop the
Dirichlet mechanism. Although its form appears quite different from existing mechanisms, they are
related through membership in a broad class of probability distributions. In particular,
the Laplacian, Gaussian, and exponential mechanisms all use distributions belonging to a parameterized
family of exponential distributions. The outputs of the Dirichlet distribution are 
equivalent to a normalized vector of i.i.d. exponential random variables, which means their
distribution also belongs to the exponential family. This connection reveals why we should expect
the Dirichlet mechanism to be well-suited to differential privacy, and the developments of
this paper formalize and confirm this intuition. 

We also point out here that the exponential mechanism is another widely used differentilly private mechanism which can be used for sensitive data ill-suited
to additive approaches~\cite{dwork2014algorithmic}. However, the exponential mechanism can be computationally demanding to implement for privacy applications with many possible outputs.
The output space here is the unit simplex, which contains uncountably many elements. The resulting
complexity of such an implementation therefore makes it infeasible~\cite{vadhan2017complexity}, especially in large dimensions, and we avoid it here. \par

An extended version of this paper which includes the proofs to the technical lemmas used in this paper can be found at \cite{extended}. \par

%The remainder of the paper is organized as follows.
%Section~\ref{pf} defines differential privacy and formulates the problems we solve.
%Then Section~\ref{body} then defines the Dirichlet mechanism and applies it to identity
%and average queries. Next, Section~\ref{sim} presents numerical results, and
%Section~\ref{con} provides concluding remarks. 

\section{Preliminaries}\label{pf}
\subsection{Notation}
In this section we establish notation used throughout the paper. We represent the real numbers by $\mathbb{R}$ and the positive reals by $\mathbb{R}_+$. For a positive integer $n$, let $[n]:=\{1,\dots,n\}$. We denote the unit simplex in $\mathbb{R}^n$ by $\Delta_n$ where 
\begin{equation*}
    \Delta_n := \left\{x \in \mathbb{R}^n \mid \sum\limits_{i=1}^n x_i = 1, x_i \geq 0 \textnormal{ for all } i\in[n]  \right\}.
\end{equation*}
We use $\Delta_n^\circ$ to represent the interior of $\Delta_n$. 
Letting $W \subseteq [n-1]$ and $\eta, \Bar{\eta} \in (0,1)$, we then define the set
\begin{equation*}
    \Delta_n^{(\eta,\Bar{\eta})} \!:=\! \left\{p\in\Delta_n^\circ \mid \sum\limits_{i\in W}p_i \leq 1-\Bar{\eta}, p_i \geq \eta \textnormal{ for all } i\in W\! \right\}\!.
\end{equation*}
%Below, privacy will be defined over~$\Delta_n^{(\eta,\Bar{\eta})}$ because its bounds on entries
%enable bounds on certain ratios of probability distributions that we consider in developing
%differential privacy. 

Letting $p$ be a vector in $\mathbb{R}^n$, we use the notation $p_{(i,j)}$ to denote the vector $(p_i,p_j)^T \in \mathbb{R}^2$, where $(\cdot)^T$ is the transpose of a vector, and $p_{-(i,j)} \in \mathbb{R}^{n-2}$ to denote the vector~$p$ with $i^{th}$ and $j^{th}$ entries removed. $\mathbb{P}[\cdot]$ denotes the probability of an event. For a random variable, $\mathbb{E}[\cdot]$ denotes its expectation and Var$[\cdot]$ denotes its variance. We use the notation $|\cdot|$ for the cardinality of a finite set. $||\cdot||_1$ denotes the 1-norm of a vector. We also use special functions 
\begin{align*}
    &\Gamma(z) := \int_0^\infty x^{z-1}\exp{(-x)}dx, & &z\in\mathbb{R}_+, \\
    &\text{beta}(a,b) := \int_0^1 t^{a-1} (1-t)^{b-1}dt, & &a,b \in \mathbb{R}_+.
\end{align*}

\subsection{Differential Privacy}
Intuitively, differential privacy guarantees that two \textit{nearby} inputs to a privacy mechanism will
generate statistically similar outputs. In differential privacy, the notion of ``nearby'' is formally defined by an adjacency relation, and we define adjacency over the unit simplex as follows. 
\begin{definition}\label{adjacency}
For a constant $b \in (0,1]$ and fixed set $ W$, two vectors $p,q \in \Delta_n^{(\eta,\Bar{\eta})}$ are said to be $b$-adjacent if there exist indices $i,j \in  W$ such that
    \begin{equation*}
        p_{-(i,j)} = q_{-(i,j)}\ \text{and} \ ||p-q||_1 \leq b.
    \end{equation*}
%    \hfill $\diamond$
\end{definition}

In words, two vectors are different if they differ in two entries by an amount not more than~$b$.
Ordinarily, differential privacy considers sensitive data differing in a single 
entry, \emph{e.g.},
one entry in a database~\cite{dwork2014algorithmic}. However, it is not possible to do so for an element of the unit
simplex because changing only a single entry would violate the condition that vectors' entries
sum to one. We therefore consider privacy with the above adjacency relation. Privacy itself
is defined next. 

\begin{definition}\label{def:dp}(Probabilistic differential privacy; \cite{machanavajjhala2008privacy})
Let~$b \in (0, 1]$ and $W \subseteq [n-1]$ be given. Fix a probability space $(\Omega,\mathcal{F},\mathbb{P})$. A mechanism $\mathcal{M}: \Delta_n^{\eta,\Bar{\eta}} \times \Omega \mapsto \Delta_n$ is said to be probabilistically $(\epsilon,\delta)$-differentially private if, for all $p \in \Delta_n^{\eta,\Bar{\eta}}$, we can partition the output space $\Delta_n$ into two disjoint sets $\Omega_1, \Omega_2$, such that
\begin{equation} \label{1req}
    \mathbb{P}[\mathcal{M}(p)\in\Omega_2]\leq\delta,
\end{equation}
and for all $q\in\Delta_n^{\eta,\Bar{\eta}}$ $b$-adjacent to $p$ and for all $x\in\Omega_1$,
\begin{equation*}
\log\left(\frac{\mathbb{P}[\mathcal{M}(p)=x]}{\mathbb{P}[\mathcal{M}(q)=x]}\right)\leq\epsilon.
\end{equation*}
%\hfill $\diamond$
\end{definition}

Probabilistic differential privacy is known to imply conventional differential 
privacy~\cite{machanavajjhala2008privacy}, and, with a slight abuse of terminology,
we refer to Definition~\ref{def:dp} simply as ``differential privacy'' for the remainder
of the paper. 

\subsection{Dirichlet Mechanism}
One contribution of this paper is to present a differentially private mechanism that, without any need of projection, maps elements of $\Delta_n$ to $\Delta_n$. In order to do so, we first introduce the Dirichlet mechanism. A Dirichlet mechanism with parameter $k \in \mathbb{R}_+$, denoted by $\mathcal{M}_D^{(k)}$, takes as input a vector $p\in\Delta_n^\circ$ and outputs $x \in \Delta_n$ according to the Dirichlet probability distribution function (PDF) centered on~$p$, \ie
\begin{equation*}
    \mathbb{P}[\mathcal{M}_D^{(k)}(p)=x] = {\frac{1}{\text{B}(kp)}} \prod\limits_{i=1}^{n-1} x_i^{kp_i-1}\left(1-\sum\limits_{i=1}^{n-1}x_i\right)^{kp_n-1},
\end{equation*}
where
\begin{equation}\label{b}
    \text{B}(kp) := \frac{\prod\limits_{i=1}^{n} \Gamma(kp_i)}{\Gamma\left(k\sum\limits_{i=1}^n p_i\right)}
\end{equation}
is the multi-variate beta function. 

We later use the parameter $k$ to adjust the trade-off that we establish between the accuracy and the privacy level of the Dirichlet mechanism. Next, we establish the privacy guarantees that
the Dirichlet mechanism provides. 
%\subsection{Identity Queries vs. Average Queries}
%This paper studies the Dirichlet mechanism in two settings which are identity queries and average query. By identity query we mean evaluating the output %of the Dirichlet mechanism when it is fed with an element of $\Delta_n$. Whereas, in average queries, we evaluate the output of the Dirichlet mechanism %when it is fed with the average of a collection of vectors, each within $\Delta_n$. In both cases, we want to show that Dirichlet mechanism satisfies %$(\epsilon,\delta)$-differential privacy with respect to the $b$-adjacency relationship that corresponds to each setting.
\section{Dirichlet Mechanism for Differential Privacy of Identity Queries} \label{body}
We begin by analyzing identity queries under the Dirichlet mechanism. 
Here, a sensitive vector~$p$ is directly input to the Dirichlet mechanism to make it approximately
indistinguishable from other adjacent sensitive vectors. 
To show the level of privacy that holds, we first bound~$\delta$, then bound~$\epsilon$.

\subsection{Computing~$\delta$} \label{ss:identity_delta}
Fix~$W \subseteq [n-1]$. 
In accordance with Definition~\ref{def:dp}, we 
partition the output space of the Dirichlet mechanism into two sets $\Omega_1, \Omega_2$ defined by
\begin{equation}\label{omega1}
    \Omega_1 := \{x\in\Delta_n \mid x_i \geq \gamma \textnormal{ for all } i \in  W\}
\end{equation}
and
\begin{equation}\label{omega2}
    \Omega_2 := \{x\in\Delta_n \mid x\notin \Omega_1\},
\end{equation}
where~$\gamma \in (0, 1)$ is a parameter that defines these sets.

Our goal is to show that the Dirichlet mechanism output belongs to $\Omega_1$ with high probability. Let $p$ be a vector in $\Delta_n^{(\eta,\Bar{\eta})}$. In the next lemma we show how to calculate $\mathbb{P}[\mathcal{M}_D^{(k)}(p)\in\Omega_1]$.

\begin{lemma}\label{lemma2}
Let $W \subseteq [n-1]$ be a given set of indices which is used to construct $\Delta_n^{(\eta,\Bar{\eta})}$,  let $p \in \Delta_n^{(\eta,\Bar{\eta})}$ and let
\begin{equation*}
    \mathcal{A}_r := \left\{x \in \mathbb{R}^{r-1} \mid \sum\limits_{i\in[r-1]}x_i \leq 1, x_i \geq \gamma
    \textnormal{ for all }i\in  W \right\},
\end{equation*}
for all $r\geq| W|+1$. Then, for a Dirichlet mechanism with parameter $k\in\mathbb{R}+$, we have that $\mathbb{P}[\mathcal{M}_D^{(k)}(p)\in\Omega_1]$ is equal to
\begin{equation*}
\dfrac{\bigintsss_{\mathcal{A}_{|W|+1}}\prod\limits_{i\in W}x_i^{kp_i-1} \! \! \left(1-\sum\limits_{i\in W}x_i\right)^{k(1-\sum\limits_{i\in W}p_i)-1} \!\! \prod\limits_{i\in W}dx_i}{\text{\emph{B}}(k\Tilde{p}_W)},
\end{equation*}
where $\Tilde{p}_W \in \Delta_{|W|+1}$ is equal to $p$ with its entries outside the set $W$ removed, and with an additional entry equal to
%\begin{equation*}
  $  1-\sum\limits_{i\in W} p_i$
%\end{equation*}
appended as its final entry. 
\end{lemma}
\begin{proof}
    Without loss of generality, take~$W = [|W|]$. 
    In order to find $\mathbb{P}[\mathcal{M}_D^{(k)}(p)\in\Omega_1]$, we need to integrate the Dirichlet PDF over the region $\mathcal{A}_n$. Therefore, we need to evaluate the $(n-1)$-fold integral
       \begin{equation}\label{CDF}\begin{aligned}
            {\dfrac{\bigintsss_{\mathcal{A}_{n}} \prod\limits_{i=1}^{n-1}x_i^{kp_i-1} \left(1-\sum\limits_{i=1}^{n-1}{x_i}\right)^{kp_n-1}dx_{n-1}\dots dx_1} {\text{B}(kp)}}.\end{aligned}
        \end{equation}
    Using a method similar to the one adopted in~\cite{rao1980incomplete}, let $y := \sum\limits_{i=1}^{n-2} x_i$. Then we can rewrite \eqref{CDF} as
    \begin{multline}\label{9}
        \frac{1}{\text{B}(kp)}\int_{\mathcal{A}_{n-2}}\int_{0}^{1-y} \prod\limits_{i=1}^{n-1}x_i^{kp_i-1} (1-y-x_{n-1})^{kp_{n}-1} \\ dx_{n-1} \dots dx_1.
    \end{multline}
    Now let $u := \frac{x_{n-1}}{1-y}$ and take the inner integral with respect to $u$. Then \eqref{9} becomes
    \begin{multline*}
         \frac{1}{\text{B}(kp)}\int_{\mathcal{A}_{n-2}}\prod\limits_{i=1}^{n-2} x_i^{kp_i-1}
         (1-y)^{k(p_{n-1}+p_n)-1}
         \int_{0}^{1}u^{kp_{n-1}-1}\\(1-u)^{kp_{n}-1}du \ dx_{n-2} \dots dx_1.
    \end{multline*}
    From the definition of the beta function, we have
    \begin{equation*}
        \int_{0}^{1}u^{kp_{n-1}-1}(1-u)^{kp_{n}-1}du = \text{beta}(kp_{n-1},kp_n).
    \end{equation*}
    Using the gamma function representation of beta functions, \ie
    \begin{equation}\label{beta}
        \text{beta}(a,b) = \frac{\Gamma(a)\Gamma(b)}{\Gamma(a+b)}, \ a,b\in\mathbb{R}_+,
    \end{equation}
    and \eqref{b}, we find that $\mathbb{P}[\mathcal{M}_D^{(k)}(p)\in\Omega_1]$ is equal to
    \begin{multline*}
         \frac{1}{\text{B}(kp)}\frac{\Gamma(kp_{n-1})\Gamma(kp_n)}{\Gamma(k(p_{n-1}+p_{n}))}
         \int_{\mathcal{A}_{n-2}} \prod\limits_{r=1}^{n-2}x_r^{kp_r-1}\\ \left(1-\sum\limits_{l=1}^{n-2}{x_l}\right)^{k(p_{n-1}+p_n)-1}dx_{n-2}\dots dx_1.
    \end{multline*}
    Using the same trick, for the next step, let $y := \sum\limits_{l=1}^{n-3} x_l$ and $u := \frac{x_{n-2}}{1-y}$. Then $\mathbb{P}[\mathcal{M}_D^{(k)}(p)\in\Omega_1]$ is equal to
    \begin{multline*}
        \frac{1}{\text{B}(kp)}\frac{\Gamma(kp_{n-2})\Gamma(kp_{n-1})\Gamma(kp_n)}{\Gamma(k(p_{n-2}+p_{n-1}+p_{n}))}\int_{\mathcal{A}_{n-3}} \prod\limits_{r=1}^{n-3}x_r^{kp_r-1}\\
        \left(1-\sum\limits_{l=1}^{n-3}{x_l}\right)^{k(p_{n-2}+p_{n-1}+p_n)-1}dx_{n-3}\dots dx_1.
        \end{multline*}
    We continue to adopt the same change of variable strategy until we are left with an integral over the region $\mathcal{A}_{| W|+1}$, which concludes the proof. 
    \end{proof} 
  
    In the previous lemma we showed how to calculate $\mathbb{P}[\mathcal{M}_D^{(k)}(p)\in\Omega_1]$. In particular we showed that instead of an $(n-1)$-fold integral of the Dirichlet PDF, the computations can be reduced to a $|W|$-fold integral. However, the expression still depends on the input vector $p$, which is undesirable and generally incompatible with differential privacy. The reason is that $(\epsilon,\delta)$-differential privacy is a guarantee for all adjacent input data and not for a specific data point. In the next lemma, we show that $\mathbb{P}[\mathcal{M}_D^{(k)}(p)\in\Omega_1]$ is a log-concave function of $p$ over the region $\Delta_n^{(\eta,\Bar{\eta})}$, which we will use to derive
    a bound for~$\delta$ that holds for all~$p$ of interest. 

\begin{lemma}\label{lemma3}
Let $W$ be a given set of indices which is used to construct $\Delta_n^{(\eta,\Bar{\eta})}$ and let $\mathcal{M}^{(k)}_D$ be the Dirichlet mechanism with parameter $k$ and input $p$. Then, $\mathbb{P}[\mathcal{M}_D^{(k)}(p)\in\Omega_1]$ is a log-concave function of $p$ over the domain $\Delta_n^{(\eta,\Bar{\eta})}$.
\end{lemma}

The proof of this lemma may be found in the extended version at \cite{extended}. 
Revisiting the definition of $\Omega_1,\Omega_2$ in \eqref{omega1} and \eqref{omega2}, we find that 
\begin{align}\label{delta}
    \mathbb{P}[\mathcal{M}_D^{(k)}(p)\in\Omega_2] &= 1-\mathbb{P}[\mathcal{M}_D^{(k)}(p)\in\Omega_1]\\
    &\leq 1- \min\limits_{p} \mathbb{P}[\mathcal{M}_D^{(k)}(p)\in\Omega_1]\nonumber =\delta.
\end{align}
From this, we see that bounding~$\delta$ can be by
minimizing~$\mathbb{P}[\mathcal{M}_D^{(k)}(p)\in\Omega_1]$, an explicit form
of which was given in Lemma~\ref{lemma2}. 
Above, we established the log-concavity of the function that we seek to minimize. As a result, instead of minimizing $\mathbb{P}[\mathcal{M}_D^{(k)}(p)\in\Omega_1]$ over the entire continuous domain of $\Delta_n^{(\eta,\Bar{\eta})}$, we can only consider the extreme points. Note that the set of points within $\Delta_n^{(\eta,\Bar{\eta})}$ form a polyhedron with at most $|W|(|W|+1)/2$ vertices. As the minimum of an unsorted list of $n$ entries can be found in linear time, the time complexity of finding $\min \mathbb{P}[\mathcal{M}_D^{(k)}(p)\in\Omega_1]$ is $\mathcal{O}(|W|^2)$. 
This analytical bound will be further explored through numerical results
in Section~\ref{sim}. Next, we develop analogous bounds for~$\epsilon$.

\subsection{Computing~$\epsilon$}
As above, fix~$\eta, \bar{\eta} \in (0, 1)$,~$b \in (0, 1]$, and~$W \subseteq [n-1]$.
Then, for a given $k\in\mathbb{R}_+$, bounding~$\epsilon$ requires evaluating the term
\begin{equation*} \label{4}
    \log{\left(\frac{\mathbb{P}[\mathcal{M}_D^{(k)}(p) = x]}{\mathbb{P}[\mathcal{M}_D^{(k)}(q) = x]}\right)},
\end{equation*}
where $p$ and $q$ are any $b$-adjacent vectors in~$\Delta_n^{(\eta,\Bar{\eta})}$. 
Let $i,j \in W$ be the indices in which $p$ and $q$ differ. Using the definition of the Dirichlet mechanism, we find
\begin{align*}\label{epsilon}
    &\log{\left(\frac{\mathbb{P}[\mathcal{M}_D^{(k)}(p) = x]}{\mathbb{P}[\mathcal{M}_D^{(k)}(q) = x]}\right)} = \log{\left(\frac{\text{B}(kq) \prod\limits_{i=1}^n x_i^{kp_i-1}}{\text{B}(kp) \prod\limits_{i=1}^n x_i^{kq_i-1}}\right)}\nonumber \\
    &\hspace{0.6in} = \log{\left(\frac{\Gamma(kq_i)\Gamma(kq_j) x_i^{kp_i-1} x_j^{kp_j-1}}{\Gamma(kp_i)\Gamma(kp_j) x_i^{kq_i-1} x_j^{kq_j-1}}\right)}\nonumber\\
    &\hspace{0.6in} = \log{\left(\frac{\Gamma(kq_i)\Gamma(kq_j)}{\Gamma(kp_i)\Gamma(kp_j) }x_i^{k(p_i-q_i)}x_j^{k(p_j-q_j)}\right)}.
\end{align*}
Since $p$ and $q$ are $b$-adjacent, we have that $p_i + p_j = q_i + q_j$. Therefore, we can compute~$\epsilon$ by evaluating the term
%\begin{equation}\label{1}\begin{split}
%     p_i + p_j = q_i + q_j. 
% \end{split}
%\end{equation}
%Combining \eqref{4}, \eqref{epsilon} and \eqref{1},
\begin{equation}\label{7}
    \log{\left(\frac{\Gamma(kq_i)\Gamma(kq_j)}{\Gamma(kp_i)\Gamma(kp_j)} \left(\frac{x_i}{x_j}\right)^{k(p_i-q_i)}\right)}.
\end{equation}
Note that if either $x_i$ or $x_j$ goes to 0, then the term in \eqref{7} would be unbounded. 
Recalling that the indices in which $p$ and $q$ can differ at are restricted to the set $W$, we
find that the values at these indices must be bounded below by~$\eta$, and therefore the
ratios of interest remain bounded as well. 

Lemma \ref{lemma5} below will provide an explicit value of~$\epsilon$,
aided in part by the next lemma. 

\begin{lemma}\label{c}
Let $W$ be a given set of indices which is used to construct $\Delta_n^{(\eta,\Bar{\eta})}$ and let $p,q$ be any $b$-adjacent vectors in $\Delta_n^{(\eta,\Bar{\eta})}$ with their $i^{\text{th}}$ and $j^{\text{th}}$ entries different. Then, for a constant $k\in\mathbb{R}_+$, we have that
\begin{equation*}\emph{$
    \frac{\text{beta}(kq_i,kq_j)}{\text{beta}(kp_i,kp_j)} \leq\frac{\text{beta}(kq_i,k(1-\Bar{\eta}-q_i))}{\text{beta}(kp_i,k(1-\Bar{\eta}-p_i))}$}.
\end{equation*}
\end{lemma}
The proof of this lemma may be found in the extended version of this paper at \cite{extended}.
\begin{lemma}\label{lemma5}
Let $W$ be a given set of indices which is used to construct $\Delta_n^{(\eta,\Bar{\eta})}$ and $\mathcal{M}^{(k)}_D$ be a Dirichlet mechanism with parameter $k$. Then, for all $x \in \Omega_1$ we have that
\begin{align*}
    &\log{\left(\frac{\mathbb{P}[\mathcal{M}_D^{(k)}(p) = x]}{\mathbb{P}[\mathcal{M}_D^{(k)}(q) = x]}\right)} \leq\\
    &\hspace{0.5in}\log\left({\frac{\text{\emph{beta}}(k\eta,k(1-\Bar{\eta}-\eta))}{\text{\emph{beta}}(k(\eta+{\frac{b}{2}}),k(1-\Bar{\eta}-\eta-{\frac{b}{2}}))}}\right)\\ &\hspace{0.5in}+{\frac{kb}{2}}\log\left(\frac{1-(|W|-1))\gamma}{\gamma}\right),
    \end{align*}
    where the parameter~$\gamma \in (0, 1)$ takes the same value of~$\gamma$
used to compute~$\delta$ in Section~\ref{ss:identity_delta}. 
\begin{proof}
Let
\begin{equation}\label{op}
    \begin{aligned}
    &v := & &\max\limits_{p,q,x} \ \  \log \left(\frac{\Gamma(kq_i)\Gamma(kq_j)}{\Gamma(kp_i)\Gamma(kp_j)}\left(\frac{x_i}{x_j}\right)^{k(p_i-q_i)}\right)\\
    & & &\begin{aligned}
    &\text{subject to} & &|p_i-q_i|\leq \frac{b}{2},\\
    & & &p_i+p_j = q_i + q_j,\\
    & & &p_i+p_j \leq 1- \Bar{\eta},\\
    & & &p_{(i,j)}\in [\eta,1-\Bar{\eta}-\eta],\\
    & & &q_{(i,j)}\in [\eta,1-\Bar{\eta}-\eta],\\
    & & &x_{(i,j)}\in \left[\gamma,1-(|W|-1)\gamma\right],
    \end{aligned}
    \end{aligned}
\end{equation}
and let $\mathcal{C}$ denote the set of feasible points of the optimization problem in \eqref{op}; we note that the first constraint enforces adjacency, while the others
encode~$p, q \in \Delta_n^{(\eta,\Bar{\eta})}$ and~$x \in \Omega_1$. 

By sub-additivity of the maximum, we have
\begin{multline}\label{v}
     v \leq \max\limits_{p,q,x \in \mathcal{C}} \log\left({\frac{\Gamma(kq_i)\Gamma(kq_j)}{\Gamma(kp_i)\Gamma(kp_j)}}\right)+\\
     \max\limits_{p,q,x \in \mathcal{C}} \log{\left(\frac{x_i}{x_j}\right)^{k(p_i-q_i)}}.
\end{multline}
Now, with
\begin{align*}
    v_1 := \max\limits_{p,q,x \in \mathcal{C}} \log{\left(\frac{x_i}{x_j}\right)^{k(p_i-q_i)}},
\end{align*}
we find
\begin{align*}
    v_1 &\leq \max\limits_{p,q,x \in \mathcal{C}} \left|k(p_i-q_i)\right|\left|\log\left(\frac{x_i}{x_j}\right)\right|\\
    &= \frac{kb}{2}\log\left(\frac{1-(|W|-1)\gamma}{\gamma}\right).
\end{align*}

Next, let $c := p_i + p_j = q_i + q_j$ and substitute $q_j,p_j$ with $c-q_i$ and $c-p_i$ respectively. Let
\begin{equation*}
    \begin{aligned}
    &v_2 := & &\max\limits_{p_i,q_i,c} \ \ \log\left( \frac{\Gamma(kq_i)\Gamma(k(c-q_i))}{\Gamma(kp_i)\Gamma(k(c-p_i))}\right) \\
    & & &\begin{aligned}
    &\text{subject to} & & \left|p_i-q_i\right| \leq \frac{b}{2},\\
    & & & c \in [2\eta,1-\Bar{\eta}],\\
    & & & p_i \in [\eta,1-\Bar{\eta}-\eta],\\
    & & & q_i \in [\eta,1-\Bar{\eta}-\eta],
    \end{aligned}
    \end{aligned}
\end{equation*}
where the constraints again encode adjacency of~$p$ and~$q$
and their containment in~$\Delta_n^{(\eta,\Bar{\eta})}$.

Then, from Lemma \ref{c} and Equation~\eqref{beta}, we have that
\begin{equation}\label{op2}
    \begin{aligned}
    &v_2 \leq & &\max\limits_{p_i,q_i,c} \ \ \log \left( \frac{\text{beta}(kq_i,k(1-\Bar{\eta}-q_i))}{\text{beta}(kp_i,k(1-\Bar{\eta}-p_i))}\right) \\
    & & &\begin{aligned}
    &\text{subject to} & & \left|p_i-q_i\right| \leq \frac{b}{2},\\
    & & & p_i \in [\eta,1-\Bar{\eta}-\eta],\\
    & & & q_i \in [\eta,1-\Bar{\eta}-\eta].\end{aligned}
    \end{aligned}
\end{equation}
Evaluating the gradient of the objective function in the optimization problem in \eqref{op2}, it can be shown that the Karush-Kuhn-Tucker (KKT) conditions of optimality are not satisfied in the interior of the set of feasible points except for points that lie on the line $p_i = q_i$. However, since the KKT conditions are only sufficient conditions (see chapter 11 of \cite{Boyd:2004:CO:993483}), satisfying them does not imply optimality which is indeed the case here. 

Evaluating points on the boundary of the feasible region shows that KKT conditions are also not satisfied, thus, the only points remaining are $(p_i,q_i)$'s in the set
\begin{multline} \label{eq:piqiset}
\left\{\left(\eta+\frac{b}{2},\eta\right),\left(1-\Bar{\eta}-\eta-\frac{b}{2},1-\Bar{\eta}-\eta\right),\right.\\
\left.\left(\eta,\eta+\frac{b}{2}\right),\left(1-\Bar{\eta}-\eta,1-\Bar{\eta}-\eta-\frac{b}{2}\right)\right\}.
\end{multline}
Note that since~$\text{beta}(a,b) = \text{beta}(b,a)$, the points in the first row give equal positive objectives and the points in the second row have equal negative objectives. Hence, we can choose the first point in Equation~\eqref{eq:piqiset} to find 
\begin{equation}\label{v1}
    v_2 = \log\left({\frac{\text{beta}(k\eta,k(1-\Bar{\eta}-\eta))}{\text{beta}(k(\eta+{\frac{b}{2}}),k(1-\Bar{\eta}-\eta-{\frac{b}{2}}))}}\right).
\end{equation}
Substituting $v_1$ and $v_2$ in \eqref{v} concludes the proof. 
\end{proof}
\end{lemma}
\begin{figure*}
% This file was created by tikzplotlib v0.8.2.
\begin{tikzpicture}

\definecolor{color0}{rgb}{0.87843137254902,1,1}
\definecolor{color1}{rgb}{0.43921568627451,0.501960784313725,0.564705882352941}

\begin{groupplot}[group style={group name = my plots, group size=3 by 1}, width=2.5in, height = 1.4in]
\nextgroupplot[
axis background/.style={fill=white!89.80392156862746!black},
axis line style={white},
tick align=outside,
tick pos=both,
x grid style={white},
xticklabel style={
  /pgf/number format/precision=3,
  /pgf/number format/fixed},
xlabel={\(\displaystyle \delta\)},
xmajorgrids,
xmin=0.01, xmax=0.105,
xtick style={color=white!33.33333333333333!black},
y grid style={white},
ylabel={\(\displaystyle \epsilon\)},
ymajorgrids,
ymin=0, ymax=6,
ytick style={color=white!33.33333333333333!black}
]
\path [fill=color0, fill opacity=0.5, very thin]
(axis cs:0.01,2.4849)
--(axis cs:0.01,2.4849)
--(axis cs:0.015,2.45504)
--(axis cs:0.02,2.4389)
--(axis cs:0.025,2.4124)
--(axis cs:0.03,2.38771)
--(axis cs:0.035,2.36458)
--(axis cs:0.04,2.34281)
--(axis cs:0.045,2.32224)
--(axis cs:0.05,2.30275)
--(axis cs:0.055,2.28413)
--(axis cs:0.06,2.29711)
--(axis cs:0.065,2.25173)
--(axis cs:0.07,2.23417)
--(axis cs:0.075,2.21743)
--(axis cs:0.08,2.2033)
--(axis cs:0.085,2.18897)
--(axis cs:0.09,2.17515)
--(axis cs:0.095,2.16179)
--(axis cs:0.1,2.14887)
--(axis cs:0.105,2.13636)
--(axis cs:0.105,2.13636)
--(axis cs:0.105,4.66383)
--(axis cs:0.1,4.67634)
--(axis cs:0.095,4.68926)
--(axis cs:0.09,4.70262)
--(axis cs:0.085,4.71644)
--(axis cs:0.08,4.73077)
--(axis cs:0.075,4.74491)
--(axis cs:0.07,4.76164)
--(axis cs:0.065,4.7792)
--(axis cs:0.06,4.82458)
--(axis cs:0.055,4.8116)
--(axis cs:0.05,4.83022)
--(axis cs:0.045,4.84972)
--(axis cs:0.04,4.87028)
--(axis cs:0.035,4.89205)
--(axis cs:0.03,4.91518)
--(axis cs:0.025,4.93988)
--(axis cs:0.02,4.96637)
--(axis cs:0.015,4.98251)
--(axis cs:0.01,5.01238)
--cycle;

\addplot [line width=0.8400000000000001pt, color1, dashed]
table {%
0.01 2.4849
0.015 2.45504
0.02 2.4389
0.025 2.4124
0.03 2.38771
0.035 2.36458
0.04 2.34281
0.045 2.32224
0.05 2.30275
0.055 2.28413
0.06 2.29711
0.065 2.25173
0.07 2.23417
0.075 2.21743
0.08 2.2033
0.085 2.18897
0.09 2.17515
0.095 2.16179
0.1 2.14887
0.105 2.13636
};
\addplot [line width=0.8400000000000001pt, color1]
table {%
0.01 5.01238
0.015 4.98251
0.02 4.96637
0.025 4.93988
0.03 4.91518
0.035 4.89205
0.04 4.87028
0.045 4.84972
0.05 4.83022
0.055 4.8116
0.06 4.82458
0.065 4.7792
0.07 4.76164
0.075 4.74491
0.08 4.73077
0.085 4.71644
0.09 4.70262
0.095 4.68926
0.1 4.67634
0.105 4.66383
};

\nextgroupplot[
legend columns=2,
legend style={at={(0.5,1.152)}, anchor=south, draw=white!80.0!black, fill=white!100!black},
axis background/.style={fill=white!89.80392156862746!black},
axis line style={white},
tick align=outside,
tick pos=both,
x grid style={white},
xticklabel style={
  /pgf/number format/precision=3,
  /pgf/number format/fixed},
xlabel={\(\displaystyle \delta\)},
xmajorgrids,
xmin=0.01, xmax=0.105,
xtick style={color=white!33.33333333333333!black},
y grid style={white},
ymajorgrids,
ymin=0, ymax=6,
ytick style={color=white!33.33333333333333!black}
]
\addlegendentry{Approximated}

\addlegendentry{No approximation}
\path [fill=color0, fill opacity=0.5, very thin]
(axis cs:0.01,2.09428)
--(axis cs:0.01,2.09428)
--(axis cs:0.015,1.99085)
--(axis cs:0.02,1.93699)
--(axis cs:0.025,1.87384)
--(axis cs:0.03,1.82441)
--(axis cs:0.035,1.78217)
--(axis cs:0.04,1.74865)
--(axis cs:0.045,1.72793)
--(axis cs:0.05,1.70029)
--(axis cs:0.055,1.67526)
--(axis cs:0.06,1.65237)
--(axis cs:0.065,1.63128)
--(axis cs:0.07,1.6117)
--(axis cs:0.075,1.59344)
--(axis cs:0.08,1.5763)
--(axis cs:0.085,1.56016)
--(axis cs:0.09,1.54492)
--(axis cs:0.095,1.53046)
--(axis cs:0.1,1.51905)
--(axis cs:0.105,1.50467)
--(axis cs:0.105,1.50467)
--(axis cs:0.105,2.71655)
--(axis cs:0.1,2.73092)
--(axis cs:0.095,2.74233)
--(axis cs:0.09,2.75679)
--(axis cs:0.085,2.77204)
--(axis cs:0.08,2.78817)
--(axis cs:0.075,2.80532)
--(axis cs:0.07,2.82358)
--(axis cs:0.065,2.84315)
--(axis cs:0.06,2.86425)
--(axis cs:0.055,2.88714)
--(axis cs:0.05,2.91217)
--(axis cs:0.045,2.93981)
--(axis cs:0.04,2.96053)
--(axis cs:0.035,2.99405)
--(axis cs:0.03,3.03629)
--(axis cs:0.025,3.08572)
--(axis cs:0.02,3.14887)
--(axis cs:0.015,3.20273)
--(axis cs:0.01,3.30616)
--cycle;

\addplot [line width=0.8400000000000001pt, color1, dashed]
table {%
0.01 2.09428
0.015 1.99085
0.02 1.93699
0.025 1.87384
0.03 1.82441
0.035 1.78217
0.04 1.74865
0.045 1.72793
0.05 1.70029
0.055 1.67526
0.06 1.65237
0.065 1.63128
0.07 1.6117
0.075 1.59344
0.08 1.5763
0.085 1.56016
0.09 1.54492
0.095 1.53046
0.1 1.51905
0.105 1.50467
};
\addplot [line width=0.8400000000000001pt, color1]
table {%
0.01 3.30616
0.015 3.20273
0.02 3.14887
0.025 3.08572
0.03 3.03629
0.035 2.99405
0.04 2.96053
0.045 2.93981
0.05 2.91217
0.055 2.88714
0.06 2.86425
0.065 2.84315
0.07 2.82358
0.075 2.80532
0.08 2.78817
0.085 2.77204
0.09 2.75679
0.095 2.74233
0.1 2.73092
0.105 2.71655
};

\nextgroupplot[
axis background/.style={fill=white!89.80392156862746!black},
axis line style={white},
tick align=outside,
tick pos=both,
x grid style={white},
xticklabel style={
  /pgf/number format/precision=3,
  /pgf/number format/fixed},
xlabel={\(\displaystyle \delta\)},
xmajorgrids,
xmin=0.01, xmax=0.105,
xtick style={color=white!33.33333333333333!black},
y grid style={white},
ymajorgrids,
ymin=0, ymax=6,
ytick style={color=white!33.33333333333333!black}
]
\path [fill=color0, fill opacity=0.5, very thin]
(axis cs:0.01,1.52988)
--(axis cs:0.01,1.52988)
--(axis cs:0.015,1.43765)
--(axis cs:0.02,1.38963)
--(axis cs:0.025,1.3413)
--(axis cs:0.03,1.30353)
--(axis cs:0.035,1.27093)
--(axis cs:0.04,1.24891)
--(axis cs:0.045,1.21951)
--(axis cs:0.05,1.20292)
--(axis cs:0.055,1.18316)
--(axis cs:0.06,1.16504)
--(axis cs:0.065,1.14832)
--(axis cs:0.07,1.1328)
--(axis cs:0.075,1.11831)
--(axis cs:0.08,1.10386)
--(axis cs:0.085,1.09224)
--(axis cs:0.09,1.08016)
--(axis cs:0.095,1.06785)
--(axis cs:0.1,1.05618)
--(axis cs:0.105,1.04538)
--(axis cs:0.105,1.04538)
--(axis cs:0.105,1.51113)
--(axis cs:0.1,1.52193)
--(axis cs:0.095,1.5336)
--(axis cs:0.09,1.54591)
--(axis cs:0.085,1.558)
--(axis cs:0.08,1.56961)
--(axis cs:0.075,1.58406)
--(axis cs:0.07,1.59855)
--(axis cs:0.065,1.61408)
--(axis cs:0.06,1.63079)
--(axis cs:0.055,1.64891)
--(axis cs:0.05,1.66867)
--(axis cs:0.045,1.68526)
--(axis cs:0.04,1.71466)
--(axis cs:0.035,1.73669)
--(axis cs:0.03,1.76928)
--(axis cs:0.025,1.80705)
--(axis cs:0.02,1.85538)
--(axis cs:0.015,1.90341)
--(axis cs:0.01,1.99563)
--cycle;

\addplot [line width=0.8400000000000001pt, color1, dashed]
table {%
0.01 1.52988
0.015 1.43765
0.02 1.38963
0.025 1.3413
0.03 1.30353
0.035 1.27093
0.04 1.24891
0.045 1.21951
0.05 1.20292
0.055 1.18316
0.06 1.16504
0.065 1.14832
0.07 1.1328
0.075 1.11831
0.08 1.10386
0.085 1.09224
0.09 1.08016
0.095 1.06785
0.1 1.05618
0.105 1.04538
};
\addplot [line width=0.8400000000000001pt, color1]
table {%
0.01 1.99563
0.015 1.90341
0.02 1.85538
0.025 1.80705
0.03 1.76928
0.035 1.73669
0.04 1.71466
0.045 1.68526
0.05 1.66867
0.055 1.64891
0.06 1.63079
0.065 1.61408
0.07 1.59855
0.075 1.58406
0.08 1.56961
0.085 1.558
0.09 1.54591
0.095 1.5336
0.1 1.52193
0.105 1.51113
};
\end{groupplot}
\node[text width=6cm,align=center,anchor=north] at ([yshift=-8mm]my plots c1r1.south) {\captionof{subfigure}{$\eta = 0.15$, $\Bar{\eta} = 0.15$, $k=6.7$ \label{subplot:one}}};
\node[text width=6cm,align=center,anchor=north] at ([yshift=-8mm]my plots c2r1.south) {\captionof{subfigure}{$\eta = 0.20$, $\Bar{\eta} = 0.20$, $k=5.1$\label{subplot:two}}};
\node[text width=6cm,align=center,anchor=north] at ([yshift=-8mm]my plots c3r1.south) {\captionof{subfigure}{$\eta = 0.25$, $\Bar{\eta} = 0.25$, $k=4.1$\label{subplot:three}}};

\end{tikzpicture}
\caption{An example where $|W|=3$ to compare the approximated and original values of $\epsilon$. At each level of $\delta$, first $\gamma$ is optimized according to the optimization problem in \eqref{op4}, then the optimal $\gamma$ is substituted in the expressions for the original and approximated values.}\label{approxfig}
\end{figure*}

We now state the main theorem of this section, which formally establishes the $(\epsilon,\delta)$-differential privacy of the Dirichlet mechanism for identity queries.
\begin{theorem}\label{theorem1}
Fix~$\eta, \bar{\eta} \in (0, 1)$ and $b \in (0, 1]$, and consider~$b$-adjacent 
vectors~$p, q \in \Delta_n^{(\eta,\Bar{\eta})}$. 
Let~$W \subseteq [n-1]$ be a given set of indices which is used to construct $\Delta_n^{(\eta,\Bar{\eta})}$. Then the Dirichlet mechanism with parameter $k\in\mathbb{R}_+$ is $(\epsilon,\delta)$-differentially private, where
\begin{multline*}
    \epsilon =\log{\left(\frac{\text{\emph{beta}}(k\eta,k(1-\Bar{\eta}-\eta))}{\text{\emph{beta}}(k(\eta+{\frac{b}{2}}),k(1-\Bar{\eta}-\eta-{\frac{b}{2}}))}\right)}+\\
    {\frac{kb}{2}}\log\left(\frac{1-(|W|-1)\gamma}{\gamma}\right),
\end{multline*}
and
\begin{equation*}
    \delta = 1- \min\limits_{p} \mathbb{P}[\mathcal{M}_D^{(k)}(p)\in\Omega_1].
\end{equation*}
\begin{proof}
    The expression for $\epsilon$ results immediately from Lemma \ref{lemma5} and the expression for $\delta$ is a direct result of \eqref{delta}.
\end{proof}
\end{theorem}
The expression given for $\epsilon$ in Theorem \ref{theorem1} contains a ratio of beta functions. In the following lemma we present upper and lower bounds for beta functions in terms of simpler functions to provide a simplified upper bound for $\epsilon$.

\begin{lemma}\label{approx}
Let $a,b > 1$. Then
\emph{\begin{equation}
    \exp{(2-a-b)}\leq \text{beta}(a,b) \leq \frac{a+b-1}{(2a-1)(2b-1)}.
\end{equation}}
\end{lemma}
\emph{Proof:} 
See~\cite{extended}. \hfill $\blacksquare$

Lemma~\ref{approx} offers a straightforward simplification of Theorem~\ref{theorem1},
though due to space restrictions we evaluate its accuracy numerically. 
\begin{remark}\theoremstyle{remark}
Note that if a mechanism is $\epsilon_1$-differentially private, it is also $\epsilon_2$-differentially private for all $\epsilon_2 \geq \epsilon_1$. Therefore, if the upper bound for $\epsilon$ after simplification of beta functions is still within the acceptable range, \emph{e.g.,} $\delta \leq 0.05$ and $\epsilon \leq 5$ \cite{mcsherry2011differentially,bonomi2012privacy,hsu2014differential}, then using
the approximate value of~$\epsilon$ does not substantially harm interpretation
of the Dirichlet mechanism's protections. 
In Figure \ref{approxfig}, for three instances of $(\eta,\Bar{\eta},k)$ and $b=0.1$, we show how the approximation captures the behavior of $\epsilon$. All three cases show that
the approximation causes an offset to the exact value of~$\epsilon$, and the level of offset decreases with the value of the original $\epsilon$.
\end{remark}

Next, we point out that the parameter~$\gamma$, which is used in the definition of $\Omega_2$, is not a parameter of the mechanism, in the sense that changing~$\gamma$ 
does not change the mechanism itself. Instead, $\gamma$ balances the trade-off between privacy level and the probability of failing to guarantee that privacy level,
\emph{i.e.,} changing~$\gamma$ can decrease~$\epsilon$ in exchange for increasing~$\delta$
and vice versa. 

In some cases, we are given the highest probability of privacy failure,~$\delta$, that is acceptable, and one must maximize the level of privacy subject to that upper bound. Let $\hat{\delta}$ denote maximum admissible value of~$\delta$. Then we are interested in minimizing $\epsilon$ while obeying $\delta \leq \hat{\delta}$. Using Theorem \ref{theorem1} to substitute $\epsilon$, let $V$ be the set of vertices of $\Delta_n^{(\eta,\Bar{\eta})}$. Then we must solve
\begin{equation}\label{op4}
    \begin{aligned}
    &\min\limits_{\gamma} \ \ \gamma \\
    &\begin{aligned}
    &\text{subject to} & & \mathbb{P}[\mathcal{M}_D^{(k)}(p)\in\Omega_1]\geq 1-\hat{\delta} \ \textnormal{for all } p \in V.\\
    \end{aligned}
\end{aligned}
\end{equation}
Note that the constraint set of the optimization problem \eqref{op4} form a convex set as the function $\mathbb{P}[\mathcal{M}_D^{(k)}(p)\in\Omega_1]$ is a strictly decreasing function of $\gamma$. Therefore, $\epsilon$ can be optimized for a given $\hat{\delta}$ using off-the-shelf convex optimization tool-boxes, and this will be done below
in Section~\ref{sim}. Next, we apply the Dirichlet mechanism to average
queries. 

\section{Dirichlet Mechanism for Differential Privacy of Average Queries}\label{body2}
In this section we consider a collection of $N$ vectors
indexed over~$i \in [N]$, with the~$i^{th}$ denoted $p^i \in \Delta_n^\circ$. The goal is to compute the average of the collection $\{p^i\}_{i\in[N]}$ while providing
differential privacy. 
We first re-define the adjacency relationship for 
the average query setting.
\begin{definition}\label{aveadj}
Fix a scalar $b \in (0, 1]$. Two collections $\{p^i\}_{i\in[N]}$ and $\{q^i\}_{i\in[N]}$ are adjacent if there is some $j$ such that
\begin{enumerate}
    \item $p^i = q^i$ for all $j\neq i$,
    \item there exist $m$ and $l$ such that $p^j_{-(m,l)}=q^j_{-(m,l)}$ and $||p^j_{(m,l)}-q^j_{-(m,l)}||\leq b$.
    %\begin{equation*}
     %   ||p^j_{(m,l)}-q^j_{-(m,l)}||\leq b.
    %\end{equation*}
\end{enumerate}
%\hfill $\diamond$
\end{definition}
As mentioned earlier, the query we now consider is the average. Set $\mathcal{P} = \{p^i\}_{i\in[N]}$ and $\mathcal{Q} = \{q^i\}_{i\in[N]}$. Mathematically we write
\begin{equation}\label{average}
    \mathcal{A(P)} := \frac{1}{N} \sum\limits_{i = 1}^{N} p^i,
\end{equation}
with $\mathcal{A(Q)}$ defined analogously. 
The next theorem formalizes the privacy protections of the Dirichlet mechanism
when applied to such averages. 

\begin{theorem}\label{theorem2}
Fix~$\eta, \bar{\eta} \in (0, 1)$ and~$b \in (0, 1]$. 
Let $W \subseteq [n-1]$ be a given set of indices which is used to construct $\Delta_n^{(\eta,\Bar{\eta})}$, let $\mathcal{P} = \{p^i\}_{i \in [N]}$ be a collection of $N$-vectors within $\Delta_n^{(\eta,\Bar{\eta})}$, let $\mathcal{A(P)}$ be the average of the collection, and let~$\mathcal{Q} = \{q^i\}_{i \in [N]}$ be adjacent
to~$\mathcal{P}$. 
Then the Dirichlet mechanism with 
parameter~$k \in \mathbb{R}_{+}$ and input $\mathcal{A(P)}$ is $(\epsilon,\delta)$-differentially private, where
\begin{multline}\label{epsilonave}
    \epsilon = \log{\left(\frac{\text{\emph{beta}}(k\eta,k(1-\Bar{\eta}-\eta))}{\text{\emph{beta}}(k(\eta+{\frac{b}{2n}}),k(1-\Bar{\eta}-\eta-{\frac{b}{2n}}))}\right)}+ \\
          \frac{kb}{2n}\log{\left(\frac{1-(|W|-1)\gamma}{\gamma}\right)},
\end{multline}
and
\begin{equation*}
    \delta = 1- \min\limits_{\mathcal{A(P)}}\mathbb{P}[\mathcal{M}^{(k)}_D(\mathcal{A(P)})\in\Omega_1].
\end{equation*}

\begin{proof}
    For all $i \in [n]$, let $A(p_i) := \mathcal{A(P)}_i$ and $x\in\Omega_1$. Then, we are interested in the quantity
    \begin{align}\label{18}
        \frac{\mathbb{P}[\mathcal{M}_D^{(k)}(\mathcal{A(P))}=x]}{\mathbb{P}[\mathcal{M}_D^{(k)}(\mathcal{A(Q))}=x]} = \frac{\text{B}(k\mathcal{A(Q)}) \prod\limits_{i=1}^{n}x_i^{kA(p_i)-1}}{\text{B}(k\mathcal{A(P)}) \prod\limits_{i=1}^{n}x_i^{kA(q_i)-1}}.
    \end{align}
    Based on the definition of the adjacency relationship for average queries in Definition \ref{aveadj}, $A(p)$ and $A(q)$ will differ only in their $m^{th}$ and $l^{th}$ entries. Taking the logarithm from both sides of \eqref{18} and using the same approach for the identity queries we have that
\begin{align}\label{19}
    &\log{\left(\frac{\mathbb{P}[\mathcal{M}_D^{(k)}(\mathcal{A(P))}=x]}{\mathbb{P}[\mathcal{M}_D^{(k)}(\mathcal{A(Q))}=x]}\right)} \leq \nonumber\\ &\hspace{0.35in}\max\limits_{\mathcal{A(P),A(Q)}}\log{\left(\frac{\text{B}(k\mathcal{A(Q)})}{\text{B}(k\mathcal{A(P)})}\right)} +\\
    &\hspace{0.35in}\max\limits_{\mathcal{A(P),A(Q)}}\log{\left(\frac{1-(|W|-1)\gamma}{\gamma}\right)^{k\left|A(p_m)-A(q_m)\right|}}.\nonumber
\end{align}
    Because $\mathcal{P}$ and $\mathcal{Q}$ are $b$-adjacent, and each entry of $\mathcal{A(\cdot)}$ represents the average of the vectors, we have that
    \begin{equation}\label{b2n}
        \left|A(p_m)-A(q_m)\right| \leq \frac{b}{2n}.
    \end{equation}
    Combining \eqref{19}, \eqref{b2n} and Lemma \ref{lemma5} completes the proof for the value of~$\epsilon$. For $\delta$, same approach for calculating $\delta$ in identity queries applies to average queries.
\end{proof}
\end{theorem}
\begin{remark}\theoremstyle{remark}
As seen in \eqref{epsilonave}, the level of privacy increases with the number of vectors present in the collection. In particular, $\epsilon \to 0$ as $n\to \infty$. This is consistent with the intuition that it would be harder to track each individual of a population when their data is mixed together in an act of averaging.
\end{remark}

\section{Accuracy Analysis}
We briefly analyze the accuracy of the Dirichlet mechanism by two metrics. First, in terms of the expected location of the mechanism output on the unit simplex and second in terms of the variance of the output vector.
\begin{proposition}\label{prop1}
Let $x \in \Delta_n$ be the output of a Dirichlet mechanism with input $p \in \Delta_n^\circ$ and parameter $k\in \mathbb{R}_+$. Then we have that $\mathbb{E}[x_i] = p_i$ and
\begin{equation}
    \text{\emph{Var}}[x_i] = \frac{p_i(1-p_i)}{k+1}\label{variance1}.
\end{equation}
\begin{proof} See \cite{extended}.
\end{proof}
\end{proposition}
\begin{remark} \theoremstyle{remark} As seen in \eqref{variance1} the variance of the output depends on the input data $p_i$. However, we can find the worst-case variance by maximizing the expression for the variance which occurs at $p_i = 0.5$. Hence, we have that
\begin{equation}\label{variance}
    \text{\emph{Var}}[x_i] \leq \frac{1}{4(k+1)}.
\end{equation}
\end{remark}
\section{Simulation Results}\label{sim}
\begin{figure}
     \centering
     \begin{subfigure}[b]{1\columnwidth}
         \centering
         \begin{subfigure}[b]{0.48\columnwidth}
            \centering
            \includegraphics[width=\columnwidth]{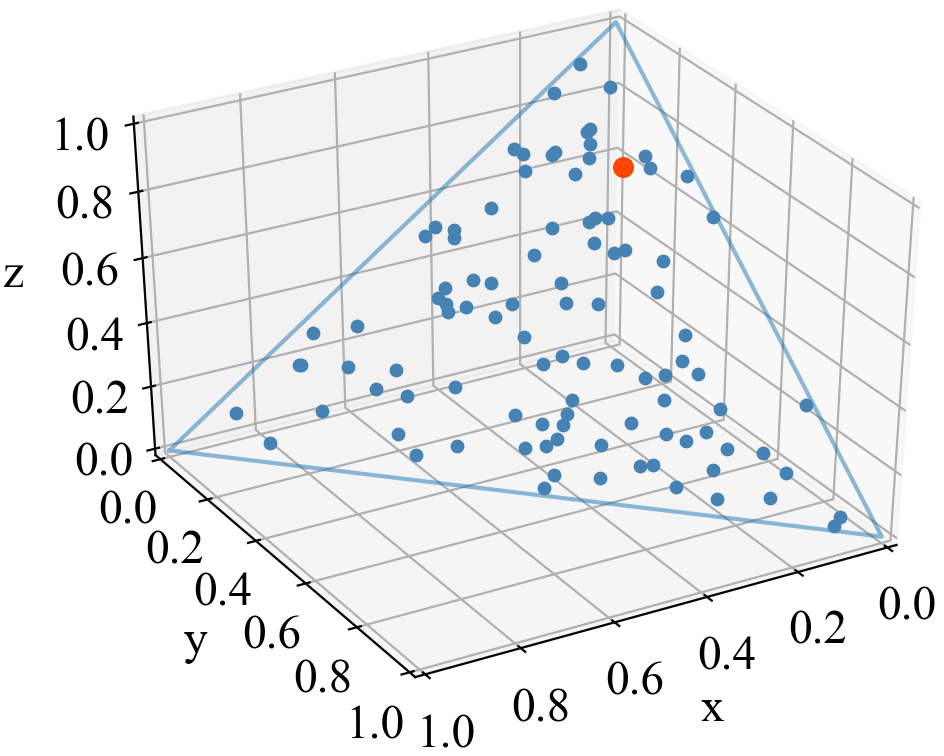}
            \label{Pcollection}
        \end{subfigure}
        \begin{subfigure}[b]{0.48\columnwidth}
            \centering
            \includegraphics[width=\columnwidth]{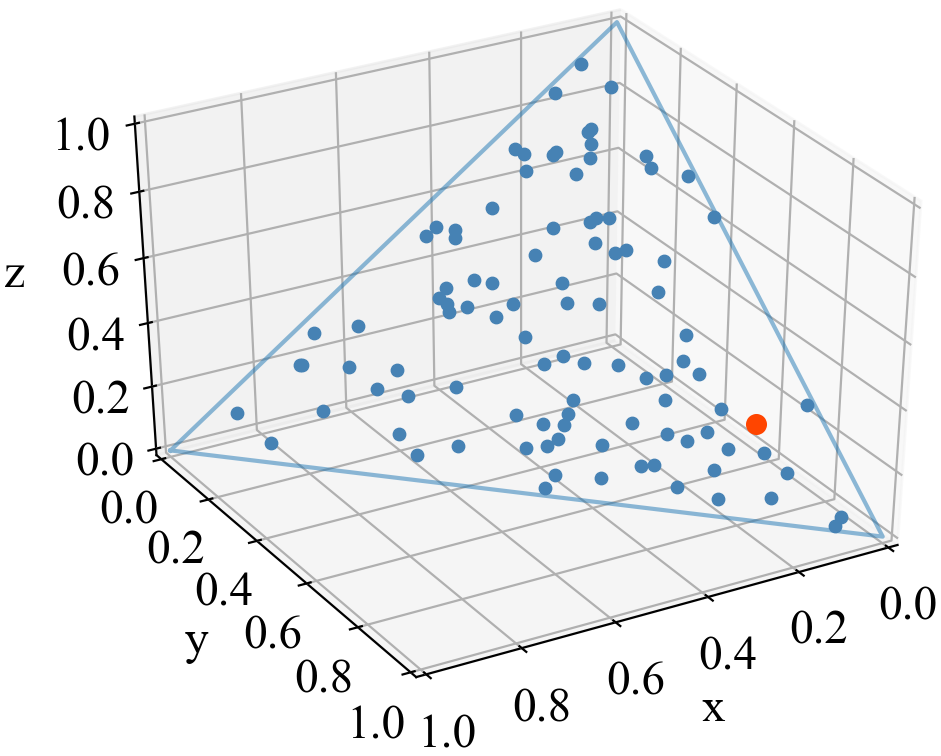}
            \label{Qcollection}
        \end{subfigure}
        \label{input}
        \caption{Visualization of two $1$-adjacent vector collections $\mathcal{P}$ and $\mathcal{Q}$. The left figure depicts $\mathcal{P}$ and the right figure corresponds to $\mathcal{Q}$. The data points with orange markers correspond to the vectors in which $\mathcal{P}$ and $\mathcal{Q}$ differ.}
     \end{subfigure}
     \hfill
     \begin{subfigure}[b]{1\columnwidth}
         \centering
         \begin{subfigure}[b]{0.48\columnwidth}
            \centering
            \includegraphics[width=\columnwidth]{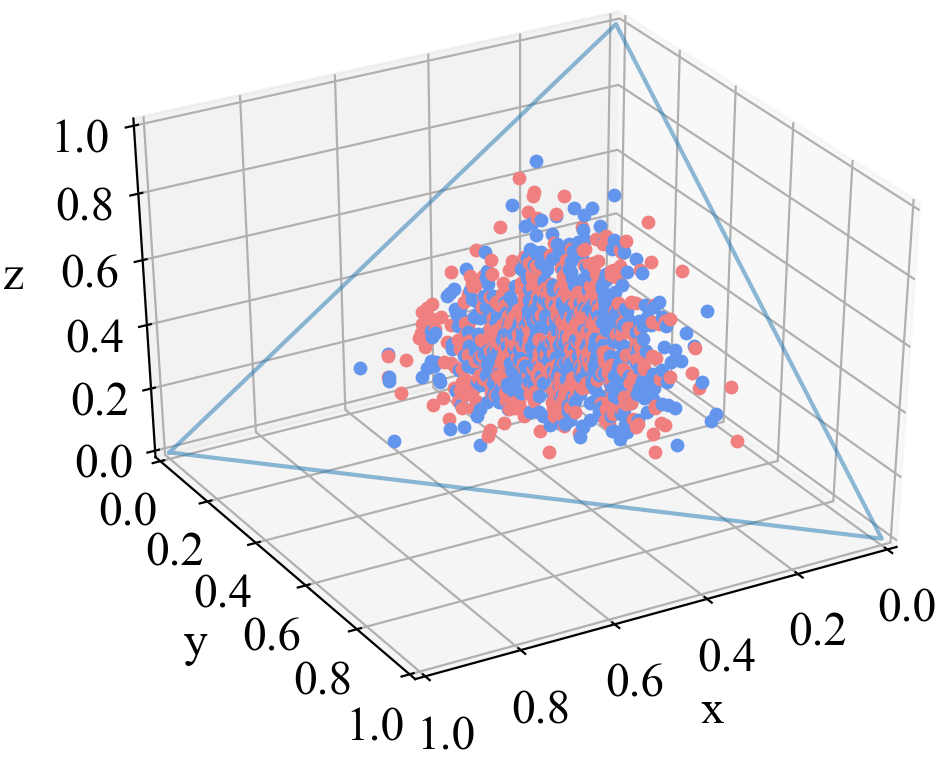}
            \label{k=10}
        \end{subfigure}
        \begin{subfigure}[b]{0.48\columnwidth}
            \centering
            \includegraphics[width=\columnwidth]{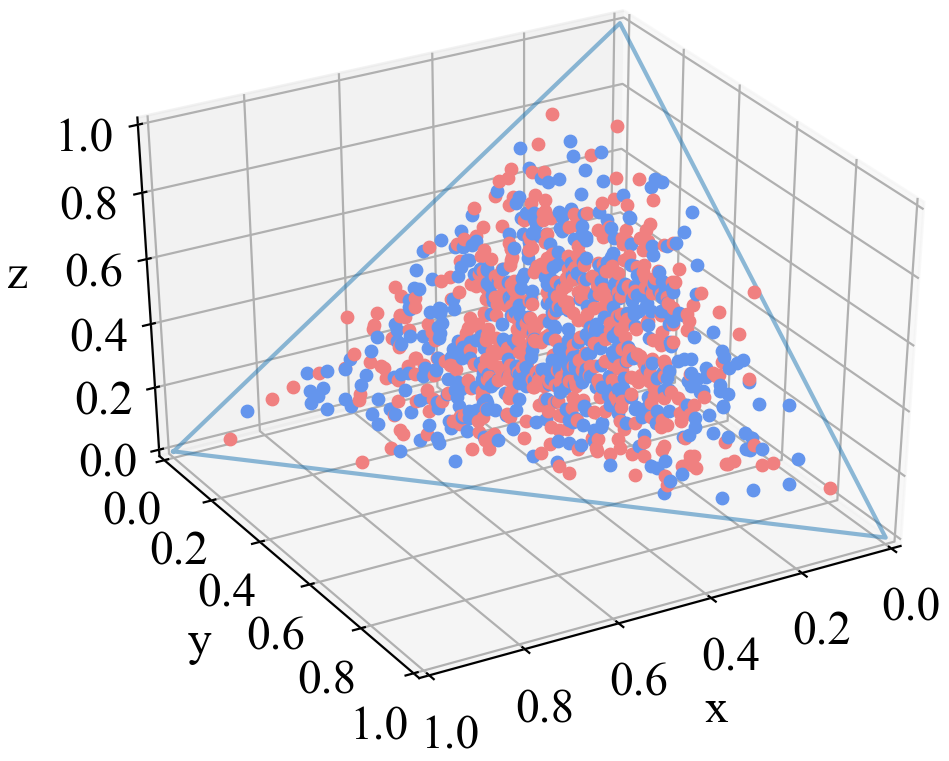}
            \label{k=24}
        \end{subfigure}
        
        \caption{The output of the Dirichlet mechanism when the input is $\mathcal{P}$ vs. $\mathcal{Q}$. The left plot shows the case where $k=24$ and the right plot corresponds to $k=10$. Each red data point corresponds to an independent run of $\mathcal{M}_D^{(k)}(\mathcal{A(P)})$ and the blue data points correspond to $\mathcal{M}_D^{(k)}(\mathcal{A(Q)})$.}\label{output}
     \end{subfigure}
        \caption{An average query on a collection of 100 vectors within $\Delta_3$.}
        \label{fig:three graphs}
\end{figure}

In this section, we simulate the output of the Dirichlet mechanism for an average query. As an example of average queries, suppose we ask a number of experts for their opinion on the probability of certain events happening, thus, a vector in the unit simplex. In order to make a decision based on all opinions, we need to integrate the opinions into one \cite{clemen1999combining}. One possible way of integrating the opinions is to take the average of the opinions. Privatizing the average of the experts' opinion while keeping their individual forecasts private is an example of average queries. \par
In Figure \ref{fig:three graphs}, we show an example of privatizing the average of 100 experts' opinions. In this example we have a collection of opinions $\mathcal{P}$, and we want to compare the output of the Dirichlet mechanism with the output of the mechanism when fed with a collection $\mathcal{Q}$ that is  1-adjacent to $\mathcal{P}$. \par
We chose $k=24$ to keep the variance of the output below $0.01$ according to \eqref{variance}. We have fixed $W$ to be the set $\{1,2\}$, $\eta = 0.05$ and $\Bar{\eta}=0.05$. Using Theorem \ref{theorem2}, for $\hat{\delta}=0.05$, we find that the mechanism is $(1.18,0.05)$-differentially private. Table \ref{statistics} shows the empirical accuracy analysis of mechanism. In Figure \ref{fig:three graphs}, we can observe that, given the location of the mechanism output, it is not possible to determine with high probability whether the input is $\mathcal{P}$ or $\mathcal{Q}$. Table \ref{statistics}, shows that we were able to achieve the desired variance.

\begin{table}[]
\centering
\begin{tabular}{@{}lc@{}}
\toprule
Statistics                                          & Values \\ \midrule
$\mathcal{A(P)}$                                    & $(0.314923,0.315923,0.320923)$\\
$\mathcal{A(Q)}$                                    & $(0.314923,0.320923,0.315923)$\\
$\hat{\mathcal{M}_D}(\mathcal{A(P)})$               & $(0.327731,0.336119,0.336149)$\\
$\hat{\mathcal{M}_D}(\mathcal{A(Q)})$               & $(0.326620,0.338976,0.334402)$\\
$\hat{\mathop{Var}}[\mathcal{M}_D(\mathcal{A(P)}]$  & $(0.00934632,0.00983122,0.0102117)$       \\
$\hat{\mathop{Var}}[\mathcal{M}_D(\mathcal{A(Q)}]$  & $(0.00922426,0.0100889,0.0100757)$       \\ \bottomrule
\end{tabular}
\caption{Comparing the average of the output of the mechanism with input collections $\mathcal{P}$ and $\mathcal{Q}$, alongside with their empirical variances. The values correspond to $k=24$.}\label{statistics}
\end{table}

In order to illustrate the effect of changing $k$ in the mechanism accuracy, Figure \ref{output} compares the output of the mechanism when $k=24$ and $k=10$. As seen in the figure, the output when $k=10$ is less concentrated around the average. It can also be seen that the probability that the output belongs to $\Omega_2$ is higher when $k=10$, which is consistent with the expressions derived in the theorems.

\section{Conclusion} \label{con}
In this work we introduced a mechanism used for privatizing data inputs that belong to the unit simplex. We used the Dirichlet distribution to probabilistically map a vector within the unit simplex to itself. We proved that the Dirichlet mechanism is differentially private with high probability in both identity and average queries. Our simulation results validated that the privacy bounds and the accuracy of the mechanism are within 
ranges typically considered in the differential privacy literature. \par
As an extension to this work, we are interested in applying the Dirichlet mechanism to privatizing a policy in a Markov decision process. In particular, we are interested in showing how accurate the Dirichlet mechanism is in terms of the total-accumulated rewards. \par
\bibliographystyle{IEEEtran}
%\bibliography{ref}

\appendix
We first state a theorem from \cite{prekopa1972} which we later use to prove Lemma \ref{lemma3}.
\begin{theorem}\label{theorem3}
Let $f_1,\dots,f_k$ be non-negative and Borel measurable functions defined on $\mathbb{R}^n$ and let
\begin{equation*}
    \begin{aligned}
    & r(t) = \underset{\lambda_1x_1+\dots+\lambda_k=t}{\text{sup}}
    & & f_1(x_1)...f_k(x_k), & t\in \mathbb{R}^n,
    \end{aligned}
\end{equation*}
where $\lambda_1,\dots,\lambda_k$ are positive constants satisfying the equality $\lambda_1+\dots+\lambda_k=1$. Then, the function $r(t)$ also Borel measurable and we have the following inequality
\begin{equation*}
    \int_{\mathbb{R}^n}r(t)dt \geq \left(\int_{\mathbb{R}^n}f_1^\frac{1}{\lambda_1}dt\right)^{\lambda_1}\dots\left(\int_{\mathbb{R}^n}f_k^\frac{1}{\lambda_k}dt\right)^{\lambda_k}.
\end{equation*}
\end{theorem}
\subsection{Proof of Lemma \ref{lemma3}}
We first review the definition of log-concave functions. A function $g: \mathbb{R}^n \mapsto \mathbb{R}$ is said to be log-concave if for all $x_1,x_2 \in \mathbb{R}^n$ and $ \theta \in [0,1]$, we have that
\begin{equation*}
    g( \theta x_1 + (1- \theta)x_2) \geq (g(x_1))^ \theta (g(x_2))^{1- \theta}.
\end{equation*}
The above condition is equivalent to
\begin{equation*}
    g(t) \geq \mathop{\text{sup}}\limits_{\theta u + (1-\theta)v = t} g(u)^\theta g(v)^{1-\theta}.
\end{equation*}
Note that function $g$ is log-concave if and only if $\log g$ is concave. Next, for $x\in \mathbb{R}^{|W|}$ and $p\in \Delta_{|W|}^{(\eta,\Bar{\eta})}$ let $f: \mathbb{R}^{|W|} \times \Delta_{|W|}^{(\eta,\Bar{\eta})} \mapsto [0,1]$ be
\begin{equation*}
    f(x,p) = \frac{\prod\limits_{i\in W}x_i^{kp_i-1}\left(1-\sum\limits_{i\in W}x_i\right)^{k(1-\sum\limits_{i\in W}p_i)-1}}{\text{\emph{B}}(kp_ W)}.
\end{equation*}
For a fixed $p\in\Delta_{|W|}^{\eta,\Bar{\eta}}$, let
\begin{equation*}
    f_1(x) := f(x,p).
\end{equation*}
Function $f_1(x)$ is the Dirichlet probability distribution function with parameter $\alpha \in \mathbb{R}^{W}$ where $\alpha:= kp_W$. Since $p \in \Delta_{|W|}^{\eta,\Bar{\eta}}$, we have that $\alpha_i\geq1$, for all $i\in\left[|W| \right]$. Therefore $f_1$ is a log-concave function \cite{Prekopa71logarithmicconcave}. Therefore,
\begin{equation}\label{ap1}
    f(t_x,p) \geq \mathop{\text{sup}}\limits_{\alpha u_x + (1-\alpha)v_x = t_x} f(u_x,p)^\alpha f(v_x,p)^{1-\alpha},
\end{equation}
for all $p \in \Delta_{|W|}^{(\eta,\Bar{\eta})}$, $t_x,u_x,v_x \in \mathbb{R}^{|W|}$ and $\alpha \in [0,1]$. Similarly, for a fixed $x\in\mathbb{R}^{|W|}$, let
\begin{equation*}
    f_2(p):=f(x,p).
\end{equation*}
Evaluating the Hessian of $\log f_2(p)$, let
\begin{equation*}
    \Bar{\psi} := \psi^{(0)}\left(k\left(1-\sum\limits_{i\in\left[|W|\right]}x_i\right)\right),
\end{equation*}
where $\psi^{(0)}$ is the digamma function. Then,
\begin{equation*}
    -\frac{\left(\nabla^2 \log f_2(p)\right)_{i,j}}{k^2} = \left\{ \begin{aligned}
        &\psi^{(0)}(kp_i)+\Bar{\psi}& &i=j\\
    &\Bar{\psi}& &i \neq j
    \end{aligned}\right.,
\end{equation*}
The digamma function is strictly increasing on interval $(0,+\infty)$. Therefore, the Hessian matrix is a sum of two negative semi definite matrices and as a result, negative semi definite itself. The aforementioned argument results in log-concavity of $f_2(p)$. Therefore,
\begin{equation}\label{ap2}
    f(x,t_p) \geq \mathop{\text{sup}}\limits_{\beta u_p + (1-\beta)v_p = t_p} f(x,u_p)^\beta f(x,v_p)^{1-\beta},
\end{equation}
for all $x \in \mathbb{R}^{|W|}$, $t_p,u_p,v_p \in \Delta_{|W|}^{(\eta,\Bar{\eta})}$ and $\beta \in [0,1]$. Let $\lambda \in [0,1]$, choose $\Tilde{u_x},\Tilde{v_x},\Tilde{u_p},\Tilde{v_p}$ such that
\begin{align*}
    &\lambda \Tilde{u_x} + (1-\lambda) \Tilde{v_x} = t_x,\\
    &\lambda \Tilde{u_p} + (1-\lambda) \Tilde{v_p} = p.
\end{align*}
Assigning $u_x$ to $x$ in \eqref{ap2}, we find 
%\begin{multline*}
%    \mathop{\text{sup}}\limits_{\alpha u_x + (1-\alpha)v_x = t_x} f(u_x,p)^\alpha f(v_x,p)^{1-\alpha} \geq \\ f(\Tilde{u_x},p)^\lambda f(\Tilde{v_x},p)^{1-\lambda},
%\end{multline*}
%and
\begin{align}\label{ap3}
    f(u_x,p)&\geq\mathop{\text{sup}}\limits_{\beta u_p + (1-\beta)v_p = p} f(u_x,u_p)^\beta f(u_x,v_p)^{1-\beta}\nonumber\\
    &\geq f(u_x,\Tilde{u_p})^\lambda f(u_x,\Tilde{v_p})^{1-\lambda}.
\end{align}
Similarly, we can write
\begin{align}\label{ap4}
    f(v_x,p)&\geq\mathop{\text{sup}}\limits_{\beta u_p + (1-\beta)v_p = p} f(v_x,u_p)^\beta f(v_x,v_p)^{1-\beta}\nonumber\\
    &\geq f(v_x,\Tilde{u_p})^\lambda f(v_x,\Tilde{v_p})^{1-\lambda}.
\end{align}
Revisiting \eqref{ap1}, using \eqref{ap3} and \eqref{ap4}, we can write
\begin{align}\label{ap5}
     &f(t_x,p) & &\geq & &\mathop{\text{sup}}\limits_{\alpha u_x + (1-\alpha)v_x = t_x} f(u_x,p)^\alpha f(v_x,p)^{1-\alpha} \nonumber\\
     & & &\geq & &\mathop{\text{sup}}\limits_{\lambda \Tilde{u_x} + (1-\lambda)\Tilde{v_x} = t_x} f(u_x,p)^\lambda f(v_x,p)^{1-\lambda}\\
    & & &\geq & &\mathop{\text{sup}}\limits_{\lambda \Tilde{u_x} + (1-\lambda)\Tilde{v_x} = t_x}f(u_x,\Tilde{u_p})^{\lambda^2}f(u_x,\Tilde{v_p})^{\lambda(1-\lambda)}\nonumber\\
    & & & & &f(v_x,\Tilde{u_p})^{(1-\lambda)\lambda} f(v_x,\Tilde{v_p})^{(1-\lambda)^2}\nonumber.
\end{align}
The second line in \eqref{ap5} is true since we have fixed $\alpha$ and the set of points satisfying the constraints is a subset of one where we are free to adjust $\alpha$. Note that
\begin{multline*}
    \lambda \Tilde{u_x} + (1-\lambda)\Tilde{v_x} = \\\lambda^2 \Tilde{u_x} + \lambda(1-\lambda) \Tilde{u_x} + \lambda(1-\lambda)\Tilde{v_x} +(1-\lambda)^2\Tilde{v_x}.
\end{multline*}
Since $\lambda^2 + \lambda(1-\lambda) + \lambda(1-\lambda) +(1-\lambda)^2=1$, Theorem \ref{theorem3} applies. Therefore, we can write
\begin{align*}\label{ap6}
    &\int_{\mathcal{A}}f(t_x,p)dt_x\geq \\ &\hspace{0.3in}\left(\int_{\mathcal{A}}f(u_x,\Tilde{u_p})du_x\right)^{\lambda^2}\left(\int_{\mathcal{A}}f(u_x,\Tilde{v_p})du_x\right)^{\lambda(1-\lambda)}\nonumber\\
    &\hspace{0.3in}\left(\int_{\mathcal{A}}f(v_x,\Tilde{u_p})dv_x\right)^{(1-\lambda)\lambda} \left(\int_{\mathcal{A}}f(v_x,\Tilde{v_p})dv_x\right)^{(1-\lambda)^2}\nonumber.
\end{align*}
By renaming the variables $t_x$, $u_x$ and $v_x$ to $x$ inside the integrals and merging the similar terms into one, we find
\begin{align*}
    &\int_{\mathcal{A}}f(x,p)dx\geq \\ &\hspace{0.3in}\left(\int_{\mathcal{A}}f(x,\Tilde{u_p})dx\right)^{\lambda}\left(\int_{\mathcal{A}}f(x,\Tilde{v_p})dx\right)^{(1-\lambda)},
\end{align*}
where $\lambda \Tilde{u_p} + (1-\lambda)\Tilde{v_p} = p$. Therefore, $\int_{\mathcal{A}}f(x,p)dx$ is log-concave which concludes the promised results. \qed
\subsection{Proof of Lemma \ref{c}}
    Let
    \begin{equation*}\begin{aligned}
        c &:=p_i+p_j \\ &= q_i+q_j.
        \end{aligned}
    \end{equation*}
    Then using \eqref{beta}, we have that
    \begin{equation}\label{20}\begin{aligned}
        \frac{\text{beta}(kp_i,kp_j)}{\text{beta}(kq_i,kq_j)} &= \frac{\Gamma(kq_i)\Gamma(k(c-q_i))}{\Gamma(kp_i)\Gamma(k(c-p_i))}\\
        & = \frac{\Gamma(kq_j)\Gamma(k(c-q_j))}{\Gamma(kp_j)\Gamma(k(c-p_j))}.
        \end{aligned}
    \end{equation}
    Using the definition of digamma function, we have
    \begin{equation}\label{24}
    \frac{\partial}{\partial x} \left[\frac{\Gamma(x-a)}{\Gamma(x-b)}\right] = \frac{\Gamma(x-a)[\psi^{(0)}(x-a)-\psi^{(0)}(x-b)]}{\Gamma(x-b)}.
\end{equation}
As the digamma function is strictly increasing on interval $(0,+\infty)$, the derivative in \eqref{24} is positive if and only if $x-b < x-a$, which is true if and only if $a<b$. Returning to \eqref{20}, we will construct an upper bound using the first identity if $q_i<p_i$ and we will construct an upper bound using the second identity if $q_j < p_j$. For correctness, suppose $q_i<p_i$. Then
    \begin{equation*}\begin{aligned}
        \frac{\text{beta}(kp_i,kp_j)}{\text{beta}(kq_i,kq_j)} &= \frac{\Gamma(kq_i)\Gamma(k(c-q_i))}{\Gamma(kp_i)\Gamma(k(c-p_i))}\\
        & \leq \frac{\text{beta}(kq_i,k(1-\Bar{\eta}-q_j))}{\text{beta}(kp_i,k(1-\Bar{\eta}-p_i))}.
        \end{aligned}
    \end{equation*}
    The other case will work identically. \qed
\subsection{Proof of Lemma \ref{approx}}
From Jensen's inequality we have that
    \begin{equation*}
        \phi \left(\int_a^b f \ dx \right) \leq \int_a^b \phi(f) \ dx,
    \end{equation*}
    where $f$ is integrable over the domain of interest and $\phi$ is a convex function. Since the exponential function is convex, we have that
    \begin{equation*}
        \begin{aligned}
        \text{beta}(a,b) &= \int_0^1 \exp{\left(\log\left(x^{a-1}(1-x)^{b-1}\right)\right)} dx \\
        &\geq \exp{\left(\int_0^1 \log \left(x^{a-1}(1-x)^{b-1}\right)dx\right)}.
        \end{aligned}
    \end{equation*}
    Evaluating the integral, we find
    \begin{equation*}
        \begin{aligned}
        &\exp{\left(\int_0^1 \log \left(x^{a-1}(1-x)^{b-1}\right)dx\right)} = \\ 
        &\int_0^1 (a-1)\log(x) + (b-1)\log(1-x)dx = 2 - (a+b).
        \end{aligned}
    \end{equation*}
    Then
    \begin{equation*}
        \text{beta}(a,b) \geq \exp(2-(a+b)).
    \end{equation*}
    The upper bound follows from the identity that
\begin{equation*}\label{square}
    2\alpha \beta \leq \alpha^2 + \beta^2,\text{ for all $\alpha,\beta \in \mathbb{R}$.}
\end{equation*}
Substituting $\alpha,\beta$ with $x^{a-1}$ and $y^{b-1}$ in the integral representation of the beta function results in the introduced upper bound. \qed

\subsection{Proof of Proposition \ref{prop1}}

    Let $\Bar{p} = \sum\limits_{r=1}^{n}kp_r$. Using equation (49.9) in \cite{kotz2004continuous} we can write
    \begin{equation*}
        \mathbb{E}[x_i] = \frac{kp_i}{\Bar{p}} =p_i,
    \end{equation*}
    and
    \begin{equation*}
        \text{Var}[x_i] = \frac{kp_i(\Bar{p}-kp_i)}{\Bar{p}^2(\Bar{p} + 1)}.
    \end{equation*}
    Since input $p$ belongs to the unit simplex, we have that $\Bar{p}=k$. Substituting $\Bar{p}$ with $k$ concludes the promised results. \qed
\end{document}